\documentclass[letterpaper,10 pt,conference]{ieeeconf}
   \usepackage[pdftex]{graphicx}
   \usepackage{placeins}
   \usepackage[flushleft]{threeparttable}
   \usepackage{tabularx}
  \usepackage{subcaption}

\usepackage[intlimits]{amsmath}
\usepackage{amsfonts,amssymb}
\DeclareSymbolFontAlphabet{\mathbb}{AMSb}
\newtheorem{theorem}{Theorem}[section]

\newtheorem{lemma}[theorem]{Lemma}
\usepackage[noabbrev]{cleveref}
\usepackage{cite}

\usepackage{algorithmic}

\usepackage{array}

\usepackage{stfloats}
\usepackage{url}

\begin{document}
%
\title{Scalable String Reconciliation by Recursive Content-Dependent Shingling}
%
%
%

\author{Bowen Song
        and Ari Trachtenberg, \{sbowen,trachten\}@bu.edu\\Department of Electrical and Computer Engineering, Boston University
}

\maketitle


\begin{abstract}
We consider the problem of reconciling similar, but remote, strings with minimum communication complexity.  This ``string reconciliation" problem is a fundamental building block for a variety of networking applications, including those that maintain large-scale distributed networks and perform remote file synchronization.  We present the novel Recursive Content-Dependent Shingling (RCDS) protocol that is computationally practical for large strings and scales linearly with the edit distance between the remote strings.  We provide comparisons to the performance of \emph{rsync}, one of the most popular file synchronization tools in active use.  Our experiments show that, with minimal engineering, RCDS outperforms the heavily optimized \emph{rsync} in reconciling release revisions for about 51\% of the 5000 top starred git repositories on GitHub.  The improvement is particularly evident for repositories that see frequent, but small, updates.
\end{abstract}


%
\IEEEpeerreviewmaketitle

\graphicspath{{Figures/}}

\section{Introduction}
%
%
%
%


%

\label{sec:motivation}

Cloud-based storage systems such as Google Drive and Dropbox provide users with extra storage space, backup and synchronization, and file sharing across different devices. These files shared on different machines inspire cooperation and allow different users to view and edit the same files concurrently. Behind the scenes, cloud storage services utilize Distributed File Systems (DFS) software such as Google File System (GFS) \cite{ghemawat2003googlefilesys}, Andrew File System (AFS) \cite{andrew_file_sys}, and Ceph \cite{weil2006ceph} to support file availability, consistency, and fault tolerance demands.

These cloud-based file systems globally push updates, whenever available, to all client machines.  However, most of the file updates are actually initiated by client devices such as personal laptops, smartphones, and tablets, which often rely on networks with interruptible or constrained bandwidth (e.g., wireless or ad-hoc networks).  When disconnected from the network or from specific peers, client devices maintain asynchronous edits in cached files, which may end up diverging from their cloud copies. The system complexity naturally grows with the number of files shared and the number of users/devices asynchronously performing edits.  The basic synchronization primitive, downloading and uploading the entire changed files, may require a (relatively) significant amount of bandwidth over an extended amount of time, even though the file updates themselves are often incremental, involving only a small number of changes at a time. It is thus advantageous to exploit the similarities between different versions of files and only communicate the file differences to perform updates and synchronizations. Indeed, an efficient string reconciliation protocol that minimizes overall communication in this manner could serve as a natural foundation for optimizing such systems.

As a further example of the underlying problem, consider two physically separated hosts, Alice and Bob, which are attempting to reconcile their local versions of a string using the minimum amount of communication.  Alice has the string ``snack," and Bob has ``snake," and their objective is for Alice to obtain just enough information from Bob to assemble the string ``snake".   The simplest solution, known as \emph{Full Sync}, involves transferring Bob's string in its entirety to Alice.  However, if Bob has some knowledge of Alice's string, perhaps he could more succinctly tell her to remove the `c' and append an `e' to her string.


The problem of string reconciliation applies to a variety of different large-scale network applications, especially those with bandwidth, consistency or availability constraints. For example, distributed file systems have to maintain file consistency between client and cloud storage and synchronize data between internal storage nodes. 

\subsection{State of the Art} 
The problem of string reconciliation is directly applicable to file synchronization which has inspired a rich body of work including~\cite{ari_bandwidth_2006,efficiently_decoding_strings_ari_2012}, which are based on shingling and ~\cite{cormode2000communication,ma2012compression,dolecek2016synchronizing,venkataramanan2013efficient}, which rely on error-correcting codes.  These protocols are typically most efficient for reconciling small strings or strings with specific structures of edits. More general algorithms such as \emph{rsync}~\cite{rsyncFromStringRecon} are capable of efficiently reconciling larger strings; however, they suffer a communication cost that scales linearly with input string length.  A similar problem also exists in protocols that use content-dependent chunking such as~\cite{low_latency_file_sync}. 


\subsection{Approach}
We approach the problem of string reconciliation by reducing it to the well-known problem of set reconciliation~\cite{cpi_set_rec_near_optimal_comm_complexity_ari}. The key to the reduction is representing a string as a \emph{multiset} of substrings, together with enough information to be unambiguously reconstructed.  A key challenge is to keep the number of symmetric differences between the \emph{multiset} of substrings close to that of the dissimilarity metric of the underlying strings. We use \emph{edit distance} metric to quantify the differences between two strings, and then design a reconciliation protocol that seeks to minimize (i) the total number of bytes transmitted, (ii) the number of rounds of communication, and (iii) the overall computation complexity of the process.  

\subsection{Contributions}
We propose a new scalable string reconciliation protocol with both an efficient communication complexity and a scalable computational complexity, two features that are not simultaneously available in existing approaches. Our main contributions are thus:
\begin{enumerate}
    \item Designing a string reconciliation solution with sub-linear communication complexity for many strings, and scalable computational complexity.
    
    \item Analysis of our approach, including  failure probability and expected/worst-case communication and computation complexities.
    
    \item Implementation as a synchronization utility with a comparison to the popular \emph{rsync} protocol.
\end{enumerate}

\subsection{Road map}
Section \ref{sec:background} includes some related work. Section \ref{sec:setsofcontent_protocol} presents our novel string reconciliation protocol. Sections \ref{sec:Analysis} and \ref{sec:evaluation} presents analysis and experiments of the protocol against different parameters and inputs. Section \ref{sec:compare_existing_work} compare the performance of our protocol with that of \emph{rsync}. Finally, Sections \ref{sec:conclusions} states our conclusion and directions of future work.

\section{Related Work}
\label{sec:background}

We describe the challenges of set/string reconciliation and explore relevant past works. In addition, we also present  bandwidth-efficient string reconciliation protocols including a detailed description of the popular \emph{rsync} algorithm~\cite{rsyncFromStringRecon}.

\subsection{Set Reconciliation}
\label{sec:set_recon}
The problem of set reconciliation serves as the foundation of our string reconciliation approach. We define the problem of set reconciliation as shown in Figure \ref{fig:setReconConcept}. Alice and Bob are hosts with sets of data $S_A$ and $S_B$ respectively. The goal is to determine the union of the two sets and transfer their symmetric differences, elements $C$ and $D$ in our figure, to achieve data consistency. The challenge is to minimize the total amount of communication, namely the total amount of bits exchanged between Alice and Bob.

\begin{figure}[!t]
    \centering
    \includegraphics[width=\linewidth]{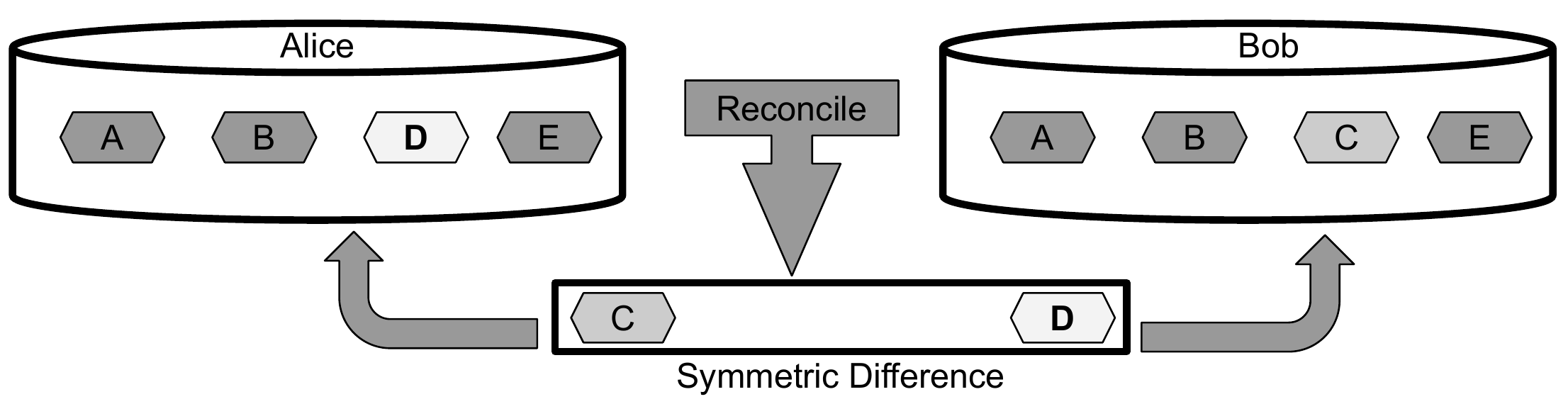}
    \caption{The problem of set reconciliation.}
    \label{fig:setReconConcept}
\end{figure}

The \emph{Characteristic Polynomial Interpolation} (CPI) \cite{cpi_set_rec_near_optimal_comm_complexity_ari} protocol encodes a host's set into a characteristic polynomial and reconciles it with a client based on evaluations of this polynomial at a collection of evaluation points. While the communication cost of CPI is linear to the number of set symmetric differences $m$, its computation complexity is cubic of $m$, which makes it only feasible to reconcile sets with a small number of symmetric differences.

\emph{Interactive CPI} \cite{practical_set_reconciliation_ari_2002} extends from \emph{CPI} to achieve linear computation complexity to $m$. It uses the divide-and-conquer technique to recursively partition the reconciling sets into disjoint space buckets, and reconciles the buckets using \emph{CPI}. The key idea is to probabilistically reduce the symmetric differences in each sub-partition until \emph{CPI} completes successfully. The full collection of reconciled differences is the union of all recovered differences.

The \emph{Invertible Bloom Lookup Table}~\cite{goodrich2011IBLT} (IBLT) is a probabilistic data structure that can be used for set reconciliation. The IBLT compacts the key-value pairs of a set of data into a table. Each set element is hashed and XOR-ed into a number of places. The IBLT can be subtracted by another one to get rid of common elements. To extract differences, the IBLT goes through a peeling process that has a high probability of retrieving all set differences. The IBLT has a communication cost linear to $m$, however, with a much larger constant than that of CPI.

\subsection{String Reconciliation}
\label{sec:str_recon}

The \emph{STRING-RECON} protocol~\cite{ari_bandwidth_2006} extracts all $k$-length substrings (k-shingles) of an input string and uses a set reconciliation protocol to exchange different shingles. It reconstructs the string from the set of reconciled shingles using an exhaustive \emph{backtracking method}. Unfortunately, the \emph{STRING-RECON} protocol has a challenging computational complexity and is not scalable to reconcile large string inputs. 

The \emph{Uniquely Decodable Shingles}~\cite{efficiently_decoding_strings_ari_2012} algorithm is an extension on \emph{STRING-RECON}, which tries to reduce the protocol's computation complexity. Upon reconciling a set of shingles, the online algorithm checks if the set of shingles can be uniquely decoded into a string by using a simple \emph{Depth First Search} and merges necessary shingles to create a uniquely decodable shingle set. The algorithm has a content-dependent characteristic based on string inputs and is not guaranteed to bring down the computational cost.

The \emph{Delta-Sync} protocol~\cite{Imrpoved_fie_sync_suel2004} recursively partitions and finds matches between two reconciling strings using hash values, then responds with a bit-map containing this information. After the matching stage, it sends unmatched data and reconstructs the string based on the bit-map. The \emph{Delta-Sync}, while suitable for reconciling large input strings, is not wildly used by the general public and does not have a sufficiently optimized implementation yet.

The \emph{rsync} algorithm~\cite{rsyncFromStringRecon} is scalable to large file sizes and is efficient for synchronizing files with great edit distance. The algorithm uses fixed partition and rolling checksum to match data between two strings and reconcile their differences. Between Two separated hosts Alice and Bob who have strings $\sigma_A$ and $\sigma_B$ respectively, Alice wishes to update her string to match Bob's $\sigma_B$. Alice would first compute rolling checksums for every $w$ sized partitions of $\sigma_A$. She sends all checksums to Bob where Bob tries to find all the matching partitions in $\sigma_B$ by calculating the checksums of all his partitions and their offsets. For every matching partition, Bob sends edit instructions for its preceding unmatched data to Alice.


Since creation, \emph{rsync} has been successfully implemented in many systems that require consistent file duplicates. \emph{Rsync} is currently used in places such as DFS client software that synchronizes files on client hosts with copies in the DFS \cite{distributed_file_sys_sync_survey}, GitLab to maintain git repositories \cite{gitlab}, and inside of Google Cloud to synchronize and migrate data between \emph{storage buckets} \cite{gsutil}. \emph{Rsync} is also a standard Linux synchronization utility included in every popular distribution serving local and remote file synchronization. We compare the performance of our protocol to that of the \emph{rsync} utility~\cite{rsyncWeb} and show the superiority of our protocol under certain circumstances in Section \ref{sec:compare_existing_work}.

\section{String Reconciliation via Recursive Content-Dependent Shingling (RCDS)}
\label{sec:setsofcontent_protocol}

Our \emph{Recursive Content-Dependent Shingling} protocol exploits content similarities between two reconciling strings. This protocol distinguishes differences between input strings by reconciling their multisets of partition hashes and transfers the partitions that are unknown to the other party. We use a \emph{local minimum chunking} technique, based on a method in~\cite{Content_dependent_Partition}, to partition the input string from non-strict local minima of content hash values. By using 64-bit hash values to represent string partitions, we create a multiset of hash shingles by concatenating hash values of two adjacent partitions and compose a de Bruijn digraph by turning shingles into vertex/edge pairs. We use a set reconciliation protocol, such as \emph{Interactive CPI} \cite{practical_set_reconciliation_ari_2002}, to reconcile the hash shingle multiset with its counterpart and then transfer substrings corresponding to the hash values unknown to the other party. Finally, we use an exhaustive \emph{Backtracking} method \cite{skiena1995reconstructing_strings_from_substrings} to uniquely decode the multiset of shingles back to a string. We recursively partition the string to create subpartitions and only transfer unmatched partitions at the bottom level of the partition tree to reduce communication cost.

\subsection{Underlying Data Structures}
\label{sec:underlying_struct}
Our protocol maintains a \emph{content-dependent partition} hierarchy based on the \emph{p-ary} tree data structure where $p$ is the maximum number of children a node can have. We build the tree from partitioning a string recursively using a \emph{local minimum chunking} method, and each partition corresponds to a node in the tree. Our protocol requires reconciling strings partitioned in a similar manner using the same parameters to increase the chance of common partitions.

\subsubsection{Local Minimum Chunking}
\label{sec:local_min_chunking}
We adapt the approach from \cite{Content_dependent_Partition} to our \emph{local minimum chunking} method to partition a given string based on its content hash values. We obtain the content hash values by a rolling hash algorithm, as described in Figure \ref{fig:rolling_hash}, which can produce pseudo-random values based on the context of the input string using a fixed window size $w$ and hash space $s$. We move the rolling hash window through the input string one character at a time and hash the $w$-length substring inside the window to a number in the space $s$. The resulting hash value array would be $N-w+1$ in length, where $N$ is the number of characters in the input string. 

\begin{figure}[!t]
    \centering
    \includegraphics[width=\linewidth]{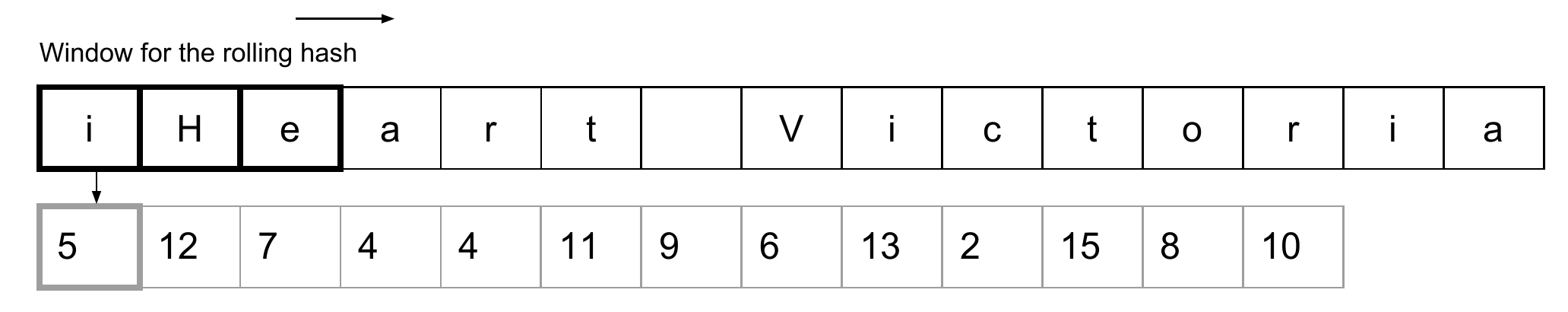}
    \caption{An example of getting content hash values using the rolling hash algorithm with window size $w=3$ and hash space $s=16$.}
    \label{fig:rolling_hash}
\end{figure}

We then partition the string based on the hash value array using a minimum inter-partition distance $h$ to control the minimum chunking size. Potential partition places include all locations of non-strict local minima considering $h$ hash values in both directions. Shown in Figure \ref{fig:content_dep_partition}, there are $3$ values in the array that satisfy as non-strict local minimum hashes for $h=2$. However, the second hash value $4$ is not a valid partition spot. Using the example hash value array, we obtain $3$ partitions for the phase ``iHeart Victoria". Besides partition decisions, $h$ also defines the upper-bound on the number of partitions of a string, $p=\left \lfloor{ \frac{N-w}{h} }\right \rfloor$.
\begin{figure}[!t]
    \centering
    \includegraphics[width=\linewidth]{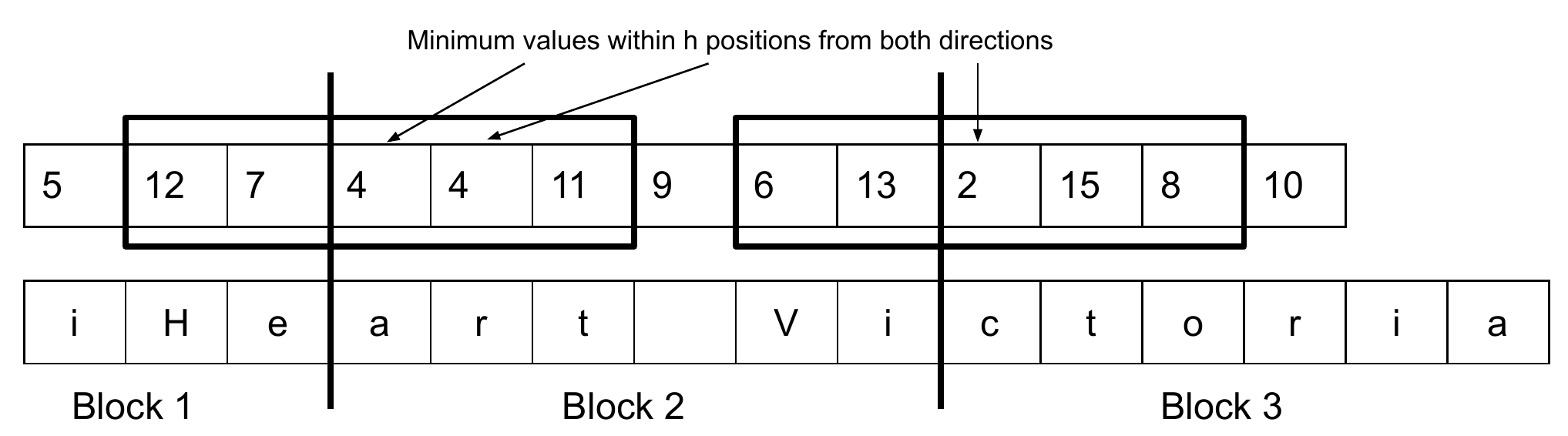}
    \caption{An example of content-dependent chunking using minimum inter-partition distance $h=2$.}
    \label{fig:content_dep_partition}
\end{figure}

\subsubsection{Hash Partition Trees}
\label{sec:hash_par_tree}
We build a hash partition tree based on string partitions from the \emph{local minimum chunking} method and hash each partition by a 64-bit hash function, $H(\cdot)$, to store in a \emph{p-ary} tree. From our previous example, we create $3$ partitions at the first level of the partition tree, shown in Figure \ref{fig:Partition_Tree}. We can grow the partition tree by creating subpartitions recursively for $L$ levels. Since we are unlikely to partition a substring further with the same parameters, we reduce the partition parameters, $s$ and $h$, at every next recursive level by $p$: $s = wp^{L-l+1}$ and $h = N/p^{l}$, where $l = \{1,2, ... , L\}$. Figure \ref{fig:Partition_Tree} shows a 2-level hash partition tree where we refer the second level partitions as \emph{terminal strings}.

\begin{figure}[!t]
    \centering
    \includegraphics[width=\linewidth]{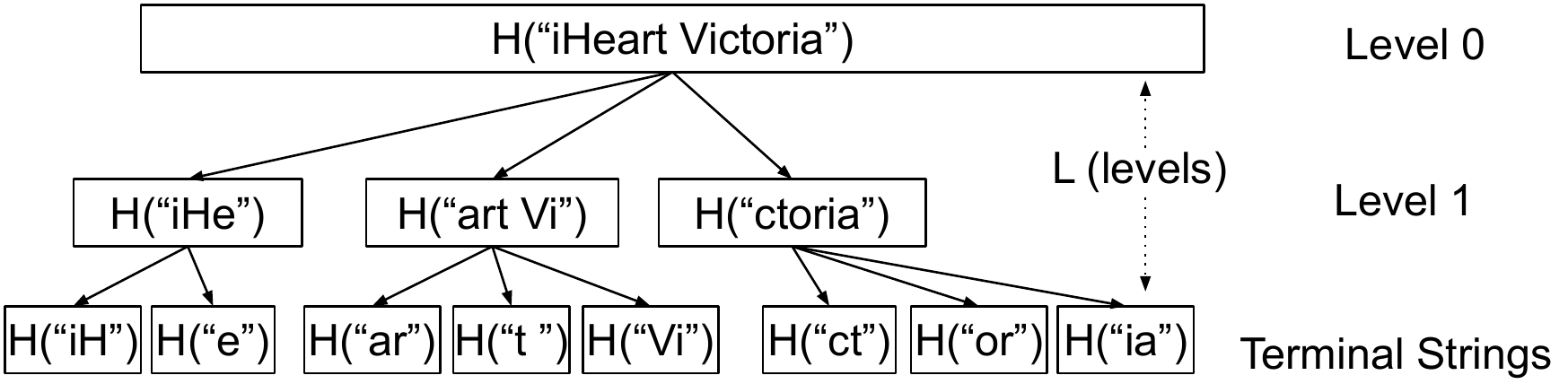}
    \caption{An example of content-dependent partition tree.}
    \label{fig:Partition_Tree}
\end{figure}

The partition tree tolerates changes to the input string without affecting most of the partitions. For example, if we change the phrase to ``Heart Victoria" with the ``i" removed from the front. According to Figure \ref{fig:rolling_hash}, the edits to the hash values array would be removing the first hash value. Figure \ref{fig:content_dep_partition} shows that this change would not affect the relative partition positions, and the resulting partitions would be ``He", ``art Vi", and ``ctoria". If we extend our recursive partition level, we also expect slight changes to the subpartitions in the first branch. Since partition decisions for other branches are isolated after the first level, the rest of the partition tree should stay the same. This property not only shows how small differences between two reconciling strings are likely to generate similar partition trees, but also allows maintaining a partition tree for incremental editing. 


\subsubsection{Converting a Partition Tree into a Shingle Set}
The last step is to transform the partition tree into a multiset of shingles before we use a set reconciliation algorithm to solve the reduced problem. We start by constructing a weighted de Bruijn digraph at every level using a 2-shingling method modified from \cite{ari_bandwidth_2006}. The edge weight describe the number of times a partition is followed by another. For each vertex, we create a shingle to capture the hash values of the previous node, the hash value of itself, and its number of occurrences.


To create a hash shingle multiset, we insert all shingles from a partition tree into a set and append their level number in the partition tree. For example, Table \ref{tab:multiset} shows the hash shingle multiset for our hash partition tree in Figure \ref{fig:Partition_Tree}. We reconcile this multiset with the one on the other reconciling host to obtain a new multiset of shingles. In order to reconstruct the string from the other host, we need to exclude the shingles not known to the other host from our multiset.

\begin{table}[!t]
    \centering
    \begin{tabularx}{\linewidth}{|X|X|X|}
    \hline
        0:H(``iHeart Victoria"):1:0 & H(``art Vi"):H(``ctoria"):1:1 & H(``iHe"):H(``art Vi"):1:1 \\
        \hline
        0:H(``iHe"):1:1 & 0:H(``iH"):1:2 & H(``iH"):H(``e"):1:2 \\ 
        \hline
        0:H(``ar"):1:2 & H(``ar"):H(``t "):1:2 & H(``t "):H(``Vi"):1:2 \\
        \hline
        0:H(``ct"):1:2 & H(``ct"):H(``or"):1:2 & H(``or"):H(``ia"):1:2 \\
        \hline
    \end{tabularx}
    \caption{An example of a multiset of hash shingles from the partition tree in Figure \ref{fig:Partition_Tree}.}
    \label{tab:multiset}
\end{table}

\subsection{Backtracking}
\label{sec:backtrack}
The backtracking starts after both reconciling hosts possess the equivalent shingle multiset. To reconstruct a string, we need to backtrack the hash shingles from the lower partition level using an exhaustive method \cite{reingold1977combinatorial_backtrack}. Our exhaustive backtracking method uses ``brute force" to trace through all possible Eulerian paths of a de Bruijn digraph in lexicographic order. Our Eulerian path follows edges with positive weights and reduces weights upon crossing until it reaches a defined depth. As an example from Figure \ref{fig:Partition_Tree}, if host Bob would like to help Alice to obtain the partition string ``iHe", he has to first compute its \emph{composition information} including the head partition hash H(``iH"), number of partitions $2$, and Eulerian tracing number $1$. Bob uses the \emph{Depth First Search} algorithm starting from the head shingle 0:H(``iH"):1:1 and searches for the next potential partitions. Bob increments the Eulerian tracing number every time he incorrectly traces to $2$ shingles. After receiving the \emph{composition information} from Bob, Alice can then trace through the shingles in the same manner.

\subsection{The Main Protocol}
\label{sec:protocol_description}
We now describe the RCDS protocol for hosts Alice and Bob who wish to reconcile their string $\sigma_a$ and $\sigma_b$ respectively. 
\begin{enumerate}\label{protocol}
    \item Hosts Alice and Bob create their \emph{hash partition} trees from $\sigma_A$ and $\sigma_B$ using the same parameters and keep dictionaries for hash-to-string conversions. \label{step:creating_partition_tree}
    
    \item Both Alice and Bob export their partition trees into two multisets of shingles $S_A$ and $S_B$ respectively. \label{step:tree_to_shingle_set}
    
    \item Using a set reconciliation protocol such as \emph{Interactive CPI}, Alice extracts the symmetrical differences between the hash shingle sets. \label{step:set_recon}
    
    \item Alice remaps her partition tree with Bob's shingles. \label{step:shingle_set_to_tree}
    
    \item Alice extracts unknown hashes by comparing against her local dictionary, and requests information from Bob. \label{step:request_unknown}
    
    \item Bob responds to the hash value inquires of \emph{terminal strings} by sending the strings as literal data. For non-terminal hash inquires, Bob finds their \emph{composition information} and send them to Alice. \label{step:answer_unkonwn}
    
    \item Alice computes hash values of the literal data for her dictionary and uses the composition information to reconstruct all other unknown partitions from the bottom up. The last string reconstruction at the top level should give Alice exactly $\sigma_b$. \label{step:reconstructing_string}
    
    \item If Bob wishes to obtain the string from Alice as well, he can go through the same procedure from step \ref{step:request_unknown} at the same time as Alice. \label{step:2_way_recon}
\end{enumerate}

\section{Analysis}
\label{sec:Analysis}

In the following analysis, we refer to all steps from Section \ref{sec:protocol_description} reconciling two similar strings with edit distance $d$. We use $L$ to denote the number of recursive partition levels, $p$ to represent the maximum of number partitions at each recursion, $h$ as the minimum inter-partition distance, and $N$ to be the larger input string size of the two reconciling strings. Since we can not control the terminal string size directly, we denote the average terminal string size for the partition trees as $T$.

\subsection{RCDS Complexities}
The following describes the typical-case in which we disregard the possibility of partition mismatch and its cascading effects. The partition mismatch happens when some number of string edits inside a partition changes the partition end-point compared to its counterpart causing the next partition to start at a different position in the string.

The communication complexity is $O(dL\alpha+dT)$ bytes, where $\alpha$ is the overhead constant for the chosen set reconciliation protocol. The main cost consists of set reconciliation from Step \ref{step:set_recon} which reconciles $O(dL)$ number of symmetric differences for the hash shingle set, and transferring of terminal strings from Step \ref{step:answer_unkonwn}.

The space complexity is $O(p^{L} + dT)$ bytes with fixed-length hash values. In Step \ref{step:creating_partition_tree}, we would create $O(p^L)$ number of partitions in a partition tree and transfer into a multiset of shingles in Step \ref{step:tree_to_shingle_set}, requiring three times the space of a partition tree in total. Then, in Step \ref{step:answer_unkonwn}, Alice receives $O(dT)$ bytes of literal data for unknown terminal hashes. 

The time complexity is $O(N\lg(h)L + dL\gamma)$ plus the set reconciliation time, where $\gamma$ is the exhaustive backtracking time reconstructing orders for partitions of each recursion. Assuming the hash computation takes constant time, each level of a $L$-level partition tree requires $O(N\lg(h))$ time to partition, moving a window of size $2h+1$ through $N-w+1$ hash value array. At each window shift, RCDS takes $O(\lg(h))$ time to compute the local minimum by maintaining a balanced Binary-Search Tree.

\subsubsection{Backtracking Time}
\label{sec:backtrack_analysis}
In Steps \ref{step:answer_unkonwn} and \ref{step:reconstructing_string}, Alice and Bob perform the similar tasks of reconstructing strings from shingles which takes $O(dL\gamma)$ time. The protocol reduces the exhaustive backtracking time to a small value, $\gamma$, in the expected case. Each string partition reconstruction would take $O(p^{D})$ time, where $D$ is the maximum degree in the de Bruijn digraph. Using Depth First Search, we require $p$ number of traversal steps, and in every step, we can have up to $D$ potential paths. The degree, $D$, corresponds to the largest number of occurrences of any partition within a partition level. Provided no duplicated partition exists, we would have the best case performance of $\gamma = p$.

In syntactic content such as programming scripts, $D$ could be a large number due to high recurrence of specific syntactic substrings. We can predict the computation cost after constructing the partition tree by determining $D$ and change the hash function for an alternative partition tree that has fewer partition duplicates if necessary.

\subsection{Partition Probability}
The characteristic of our \emph{local minimum chunking} method depends on its cutting parameters including minimum inter-partition distance $h$, rolling hash window size $w$, and content hash space $s$. We define cut-points as the two endpoints of a partition and adopt an analysis for strict maximum chunking method from \cite{Content_dependent_Partition}.

The rolling hash converts $w$ number of consecutive string characters to a hash value. The array of hash values should occupy the entire hash space and avoid biased partition decisions from reoccurring string content. In any type of strings, we would create more unique hashes by hashing longer substrings. Generally, we want $|\Sigma|^w>>s$, where $|\Sigma|$ is the size of string content alphabet, so that the intended space $s$ is filled. 

The choice of $h$ and $s$ directly affects the probability of partition, as described in Section \ref{sec:local_min_chunking}. In Lemma \ref{lemma:partition_prob}, we discuss the probability of an arbitrary point in a string to be a non-strict local minimum potential cut-point. If the probability is too high, we lose the uniqueness of each cut-point. In the worst-case scenario where all content hashes have the same value making every point to be a potential cut-point, our partition algorithm would be no different than a fixed-size partition method with a block size of $h$. On the other hand, if the probability is too small, we would have no partition at all. In general, we want $O(p)$ number of potential cut-points which is roughly equal to \emph{string size * cut-point probability}.

\begin{lemma}[\cite{Content_dependent_Partition}, Remark 57]
\label{lemma:partition_prob}
Given an independent and identically distributed (i.i.d.) number array $A$, where $A[j] \in \{0,1,\dots, s-1\}$,$\ j=\{0,1,\dots,2h\}$, and $s,h \in \mathbb{N^+}$. For any $k \in \{0,\dots,2h\}$, the probability that $A[k]\leq A[j]$ for all $k\neq j$ is given by:
\begin{equation}
    p = \sum_{j=0}^{s}\frac{1}{s}\left(\frac{j}{s}\right)^{2h}.
    \label{eq:partition_prob}
\end{equation}
\end{lemma}
\begin{proof}
We consider an arbitrary value $A[k]$ within the number array of $2h+1$ values to be a non-strict minimum value if all other numbers are equal or larger than said value. For example, if $A[k] = 0$ is a minimum value in an array of numbers in the range of $[0,s-1]$ inclusive, then all $2h$ other values would have to choose from values that are 0 or above. The probability for all $2h$ other i.i.d. values being larger than or equal to 1 is $\left(\frac{s-0}{s}\right)^{2h}$, because all of these $2h$ values can be a value from $0$ to $s-1$. Since there is an equal chance, $\frac{1}{s}$, for $A[k]$ to be any number between $0$ to $s-1$ inclusive, we sum up all the probabilities for $A[k] \in \{0,\dots,s-1\}$ while $2h$ other values being greater or equal to $A[k]$:
\begin{equation}
    \frac{1}{s}\left(\frac{s-0}{s}\right)^{2h}+\frac{1}{s}\left(\frac{s-1}{s}\right)^{2h}+\dots+\frac{1}{s}\left(\frac{s-s}{s}\right)^{2h}.\label{eq:non_strictmin_expan}
\end{equation}
\end{proof}

Since $h \approx \frac{N}{p}$ where $N$ is the string size and $p$ is the maximum number of partitions at each recursion, we use $p$ to decrease $s$ and $h$ as we move down to lower recursive partition levels. We apply the same $s$ and $h$ for all partitions in the same recursion level to avoid parameter mismatch.

\subsection{Weaknesses}
\label{sec:worst_case_analysis}
For partition trees, both partition mismatch and duplication can introduce a large amount of communication and computation cost. Partition mismatch can result from string edits or cascading effects of partition mismatch from the preceding or the upper-level partitions. As for partition duplicates, they increase string reconstruction time exponentially.

\subsubsection{Sparse String Edits}
RCDS is most effective when edits are concentrated in one location since all unmatched terminal strings must be transferred by the host. If the string edits are spread across all terminal strings, they would all have to be sent as unmatched data.

Moreover, for every unmatched non-terminal partition, we are required to compute its string \emph{composition information}. We would have to reconstruct all internal string partitions in the partition tree as well.

\subsubsection{Cascading Partition Mismatch}
\label{sec:mismatch_partition}
Partition mismatch happens where two partitions share a content relative start point but end differently due to string edits. The effect of partition mismatch can get carried over to its subsequent partitions causing cascading mismatches. Given $d$ amount of edit distance between the reconciling strings, we could have $O(dw)$ number of differences between the content hash arrays. New values could affect partition decisions if they are the strict minimum or the left most non-strict minimum value within $h$ positions of a cut point.


\subsection{Failure Modes}
\label{sec:success_rate}
The only causes of failure for RCDS are hash collisions and set reconciliation failures.

\begin{theorem}
The upper-bound probability of failure using \emph{Interactive CPI} \cite{practical_set_reconciliation_ari_2002} as the set reconciliation protocol is:
\begin{equation}
    \left[1-exp\Bigg(-\frac{(n+1)^2}{2(2^{b}+1-n)}\Bigg)\right] + \Bigg(\frac{n-1}{2^{2b}}\Bigg)^k
    \label{eq:failure_modes}
\end{equation} where $n$ is the maximum total number of partitions created by both reconciling hosts, $k$ is the CPI redundancy factor, and $b$ is the length of the hash function in bits.
\end{theorem} 
\begin{proof}
The first term in Equation \ref{eq:failure_modes} is the hash collision rate for fixed-length hash functions~\cite{hash_match_probabilities} and the second term is the failure rate of \emph{Interactive CPI}~\cite{practical_set_reconciliation_ari_2002}. We can swap these terms accordingly if we change the underlying hash function or the set reconciliation protocol. RCDS uses set reconciliation to reconcile differences between two hash shingle multisets. The hash shingles consist of all two adjacent partition hashes in a partition tree, therefore, each shingle is $2b$ bits. The total number of elements in our multiset is at most the total number of partitions $n$ in a partition tree, and the size of hash values is $b$ bits.
\end{proof}

\section{Experiments and Evaluation}
\label{sec:evaluation}
We implement~\cite{RCDS_implementation} RCDS and test against random strings, random collections of books from Project Gutenberg \cite{gutenberg}, and popular programming scripts from GitHub. We set up two reconciling hosts as a parent and a child process communicating via a socket on one piece of commodity hardware (iMac 2012 with 2.9GHz quad-core Intel Core i5 processor and 8GB of 1600MHz DDR3 memory). To control the edit distance between the two reconciling strings, we generate a similar string by randomly inserting or deleting substrings of random length from parts of the original copy controlling the overall upper-bound on the \emph{edit burst distance} to mimic human edits. In the following experiments, we include 1000 observations with $95\%$ confidence intervals.

\begin{figure*}[!t]
    \centering
    \begin{subfigure}{0.33\textwidth}
        \centering
        \includegraphics[width=\textwidth]{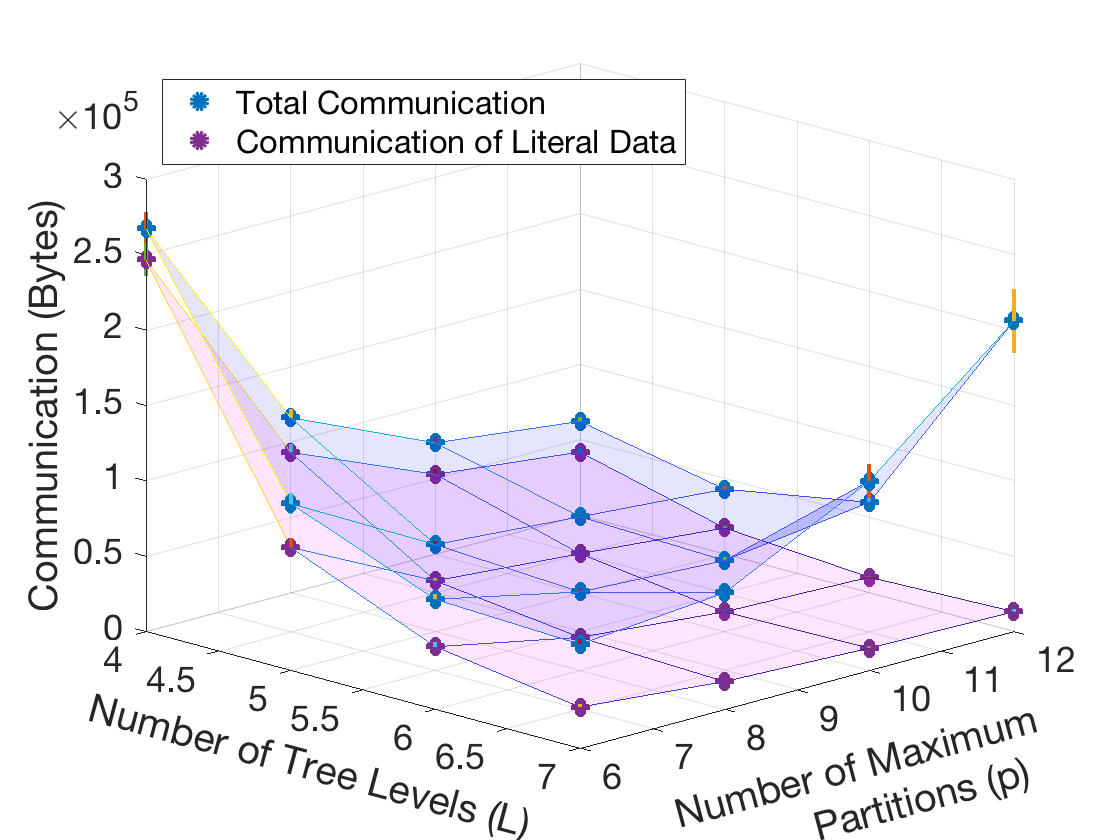}
        \caption{Communication under various parameters.}
    \label{fig:tree_comm}
    \end{subfigure}
     \begin{subfigure}{0.33\textwidth}
        \centering
        \includegraphics[width=\textwidth]{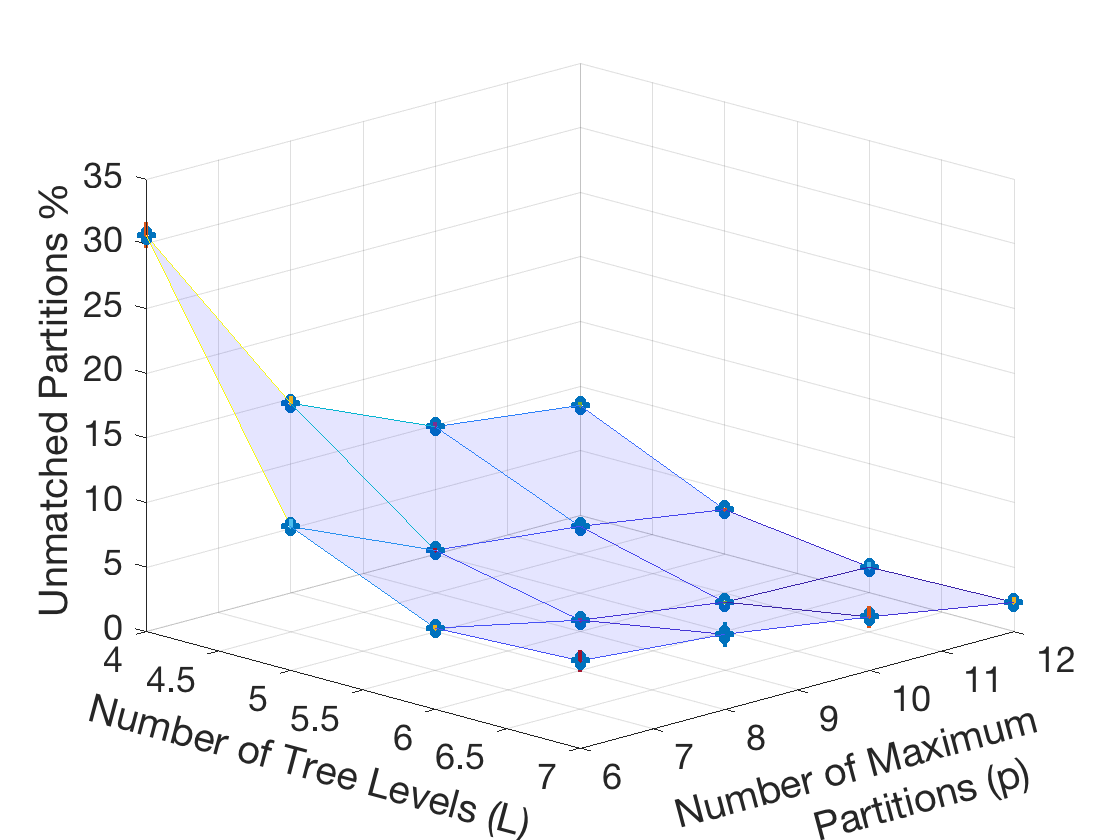}
        \caption{Percentage of unmatched tree partitions.}
        \label{fig:tree_partition_percent}
    \end{subfigure}
         \begin{subfigure}{0.32\textwidth}
        \centering
        \includegraphics[width=\textwidth]{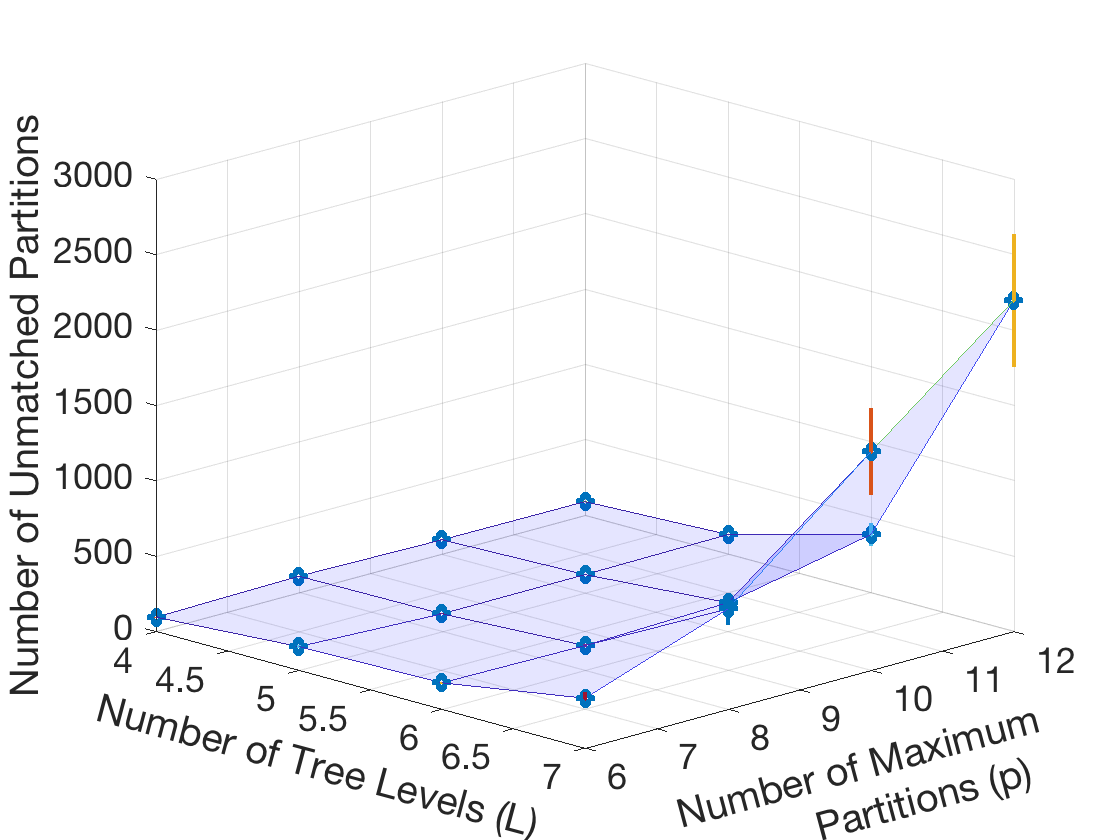}
        \caption{Number of unmatched tree partitions.}-
        \label{fig:tree_partition_num}
    \end{subfigure}
    \caption{Evaluating RCDS performance with different parameters.}
    \label{fig:diff_tree}
\end{figure*}

\subsection{Varying Protocol Parameters}
\label{sec:tuning_parameters}
We randomly select strings of $10^6$ characters and create a copy each time with an upper-bound of $10^3$ edit burst distance. We observe the performance trend by changing the partition tree parameters. The size of the partition tree grows with respect to the increase of maximum number of partitions and number of recursion levels. In Figure \ref{fig:diff_tree}, the partition tree size grows from left to right across the x and y-axis.

For RCDS communication cost (Figure \ref{fig:tree_comm}), we see a sharp decline for the amount of literal data transfer as the size of the partition tree grows. Figure \ref{fig:tree_partition_percent} shows the percentage of partition differences dropping in a similar manner. In other words, most of the partitions match with their counterparts which reduce the amount of literal data transfer. Combining the two graphs, we see the result of partition hierarchy successfully isolating string edits while matching unchanged substrings to lower the communication cost.

Unfortunately, Figure \ref{fig:tree_comm} shows a tail raise for the total communication cost as the partition tree grows, leaving the middle ground as the optimal area. This tail raise is due to set reconciling resolving more number of partition differences as shown in Figure \ref{fig:tree_partition_num}.



\begin{figure*}[!t]
    \centering
     \begin{subfigure}{0.33\textwidth}
        \centering
        \includegraphics[width=\textwidth]{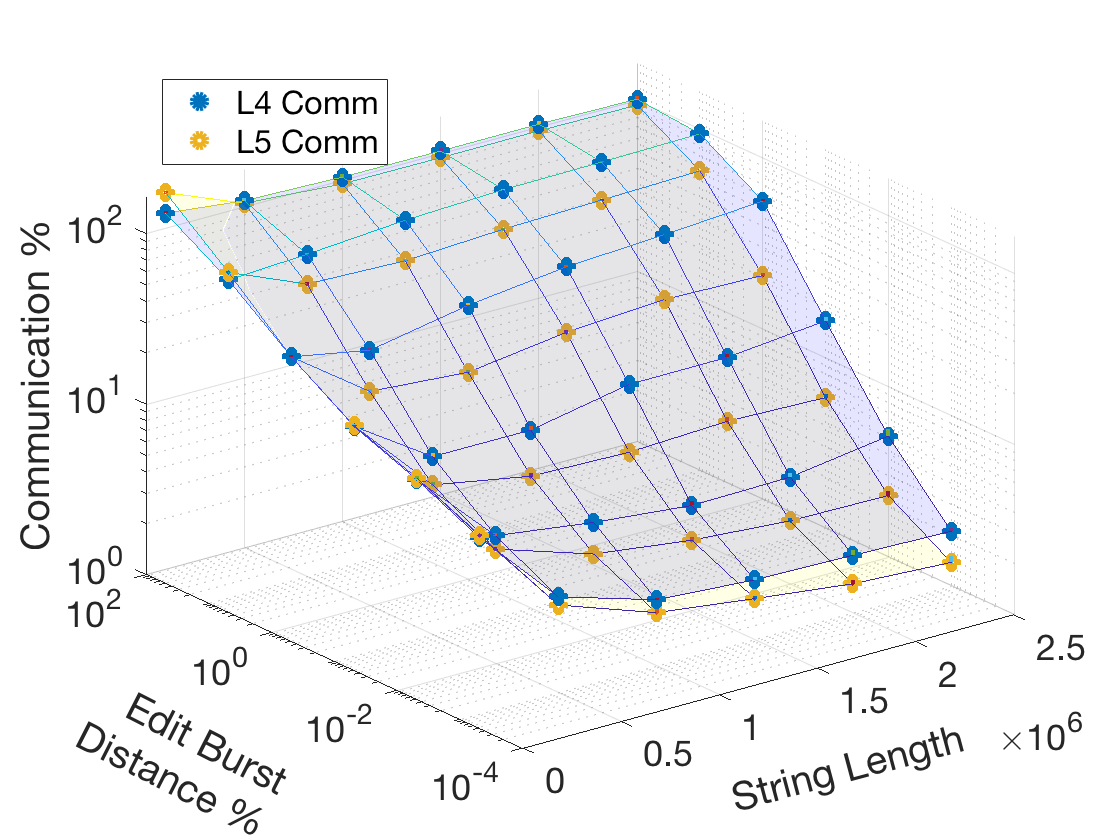}
        \caption{Communication cost percentage for different inputs.}
        \label{fig:comm_comp}
    \end{subfigure}
    \begin{subfigure}{0.33\textwidth}
        \centering
        \includegraphics[width=\textwidth]{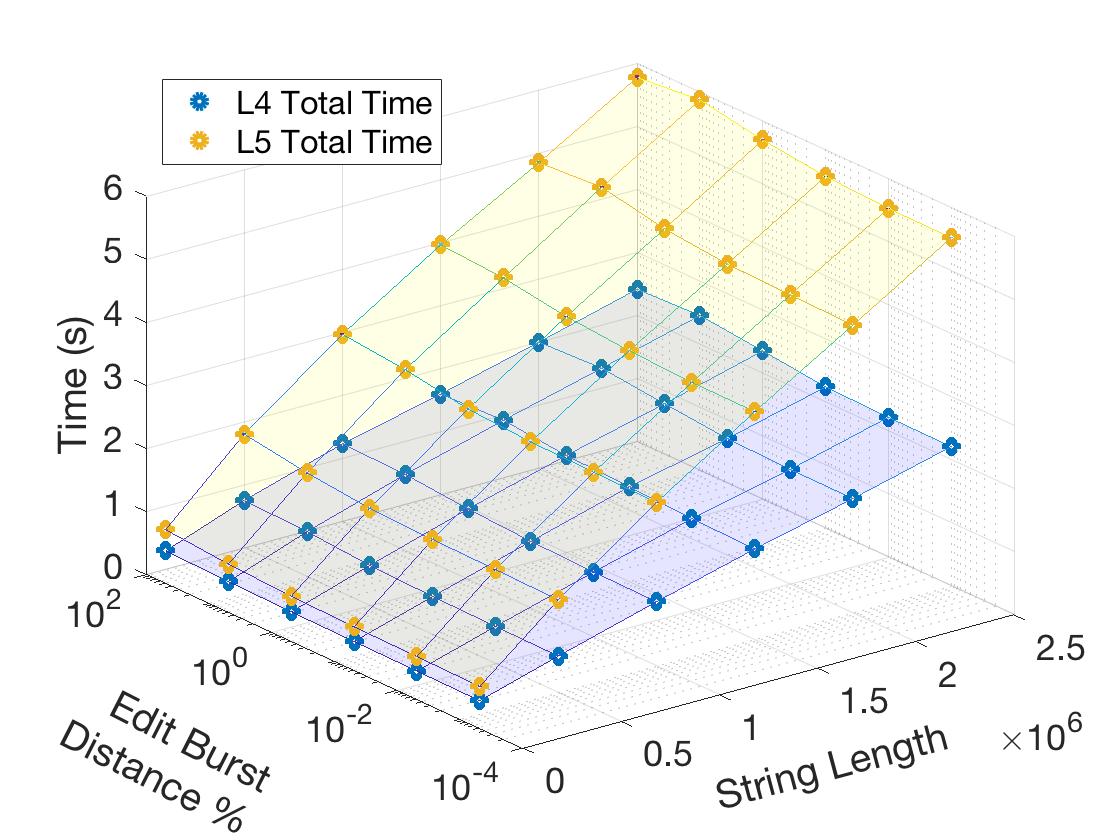}
        \caption{Time cost of different inputs. }
        \label{fig:time_comp}
    \end{subfigure}
    \begin{subfigure}{0.32\textwidth}
        \centering
        \includegraphics[width=\linewidth]{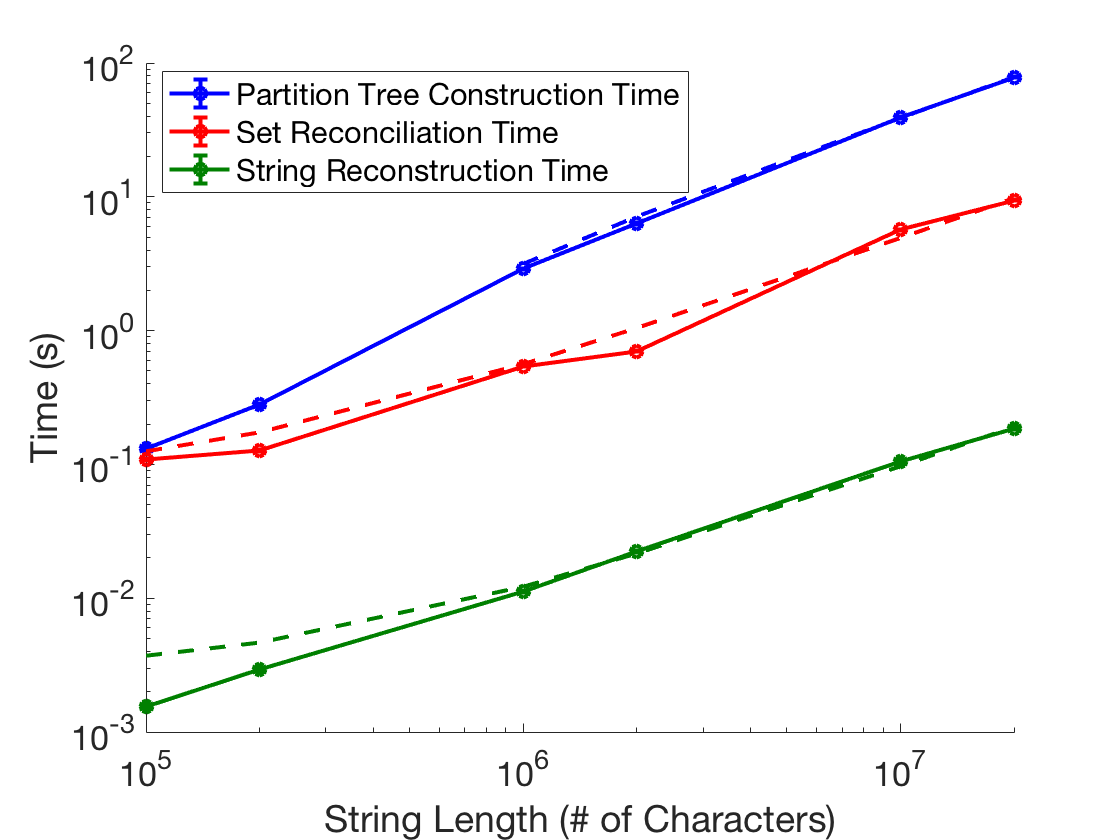}
        \caption{Time per operations with dashed linear fit.}
        \label{fig:perf_diff_file_time_fit}
    \end{subfigure}
    \caption{Evaluating RCDS performance with different input strings.}
    \label{fig:diff_input}
\end{figure*}

\subsection{Varying String Inputs}
\label{sec:par_tree_evaluation_analysis}
We compare the performance of RCDS under fixed-level partition trees while changing the reconciling string size in the range of $[10^5,\  2*10^6]$ characters and edit burst distance $[10^{-3},\ 100]$ in percentage of the original string size in Figure \ref{fig:diff_input}. We fix the number of maximum partitions at each recursion $p=8$ and levels of partition recursion $L = 4,5$.

Figure \ref{fig:comm_comp} shows the percentage of RCDS communication cost over the size of reconciling strings. For shorter strings, the protocol introduces a relatively large amount of overhead from reconciling multiset of hash shingles. The overhead becomes less significant as the size of reconciling string increases. The general trend of communication cost increases as the edit burst distance grows while staying almost constant to changes in the input string size. We also see the communication cost percentage of the 4-level partition tree well below that of the 5-level partition tree on a logarithmic scale. This is because the extra partition level helped to break partitions further and created more matching terminal partitions to reduce the communication cost.


The construction of partition trees dominates the RCDS time cost. Figure \ref{fig:time_comp} shows that more recursion levels take longer time to compute. We can reduce this time by using parallel sorting algorithms such as \cite{rsyncFromStringRecon}. Combining \Cref{fig:comm_comp,fig:time_comp}, we can see the adjustable trade-off between the communication and computation cost of RCDS.

We enlarge the range for input string length and fix the edit burst distance at $1000$ to elaborate on the time cost for each operation in Figure \ref{fig:perf_diff_file_time_fit}. The time costs for constructing partition trees, reconciling set differences, and reconstructing string from hash shingles, are plotted individually with their color-coded linear fits. The time cost of each operation generally follows a linear trend with partition tree construction dominating the cost as the input string length grows.


\section{Comparison to Rsync}
\label{sec:compare_existing_work}

We compare the performance of RCDS protocol to that of the \emph{rsync} utility (version 2.6.9) by reconciling a similar but much larger set of data using the same setup as the previous experiments in Figure \ref{fig:all_repo_perf}. The goal of this comparison is to provide a performance benchmark for two important intended use cases, being the synchronization of single files, and the synchronization of full git repositories.

\subsection{Single Files}
For single files, we show a comparison of the communication cost synchronizing files of different length under 100 random burst edits. Shown in Figure \ref{fig:vs_rsync_comm}, \emph{rsync} has a communication cost linear to the reconciling string size, whereas RCDS has a sub-linear communication cost to the increasing input string length. 
\begin{figure*}[!t]
    \centering
    \begin{subfigure}{0.33\textwidth}
    \centering
    \includegraphics[width=\linewidth]{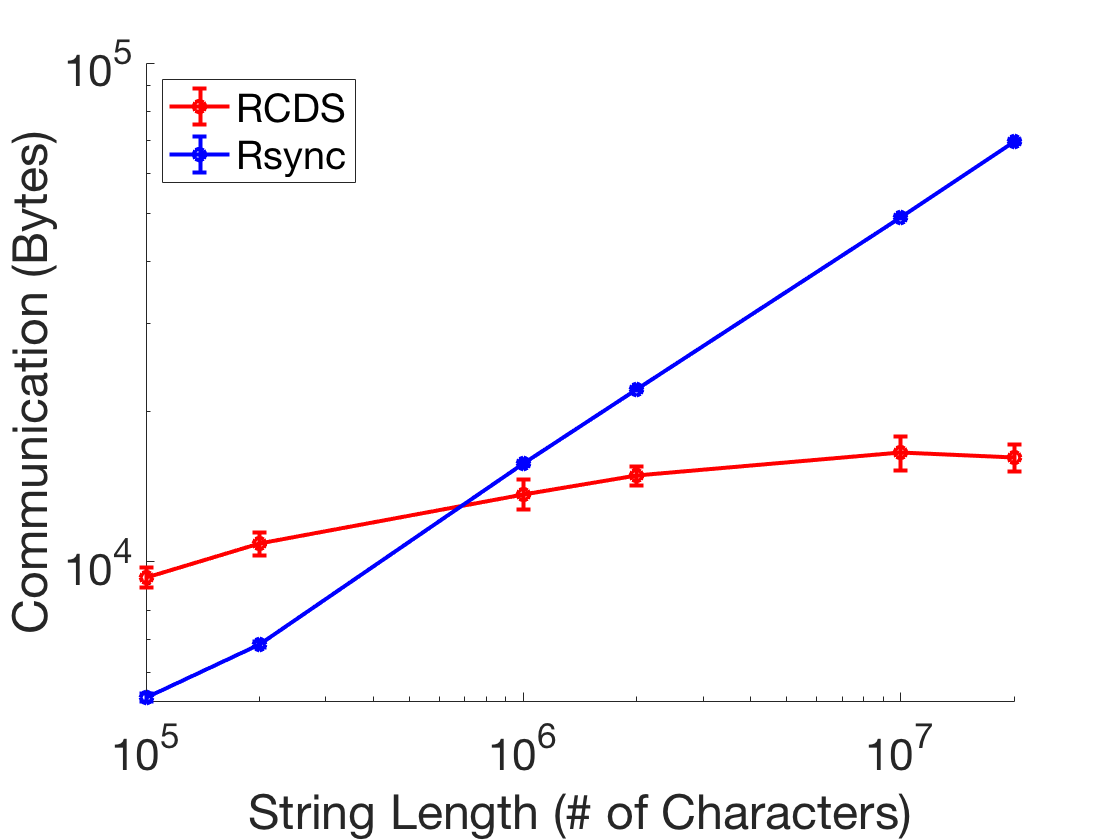}
    \caption{Reconciling single files.}
    \label{fig:vs_rsync_comm}
    \end{subfigure}
    \begin{subfigure}{0.32\textwidth}
        \centering
        \includegraphics[width=\linewidth]{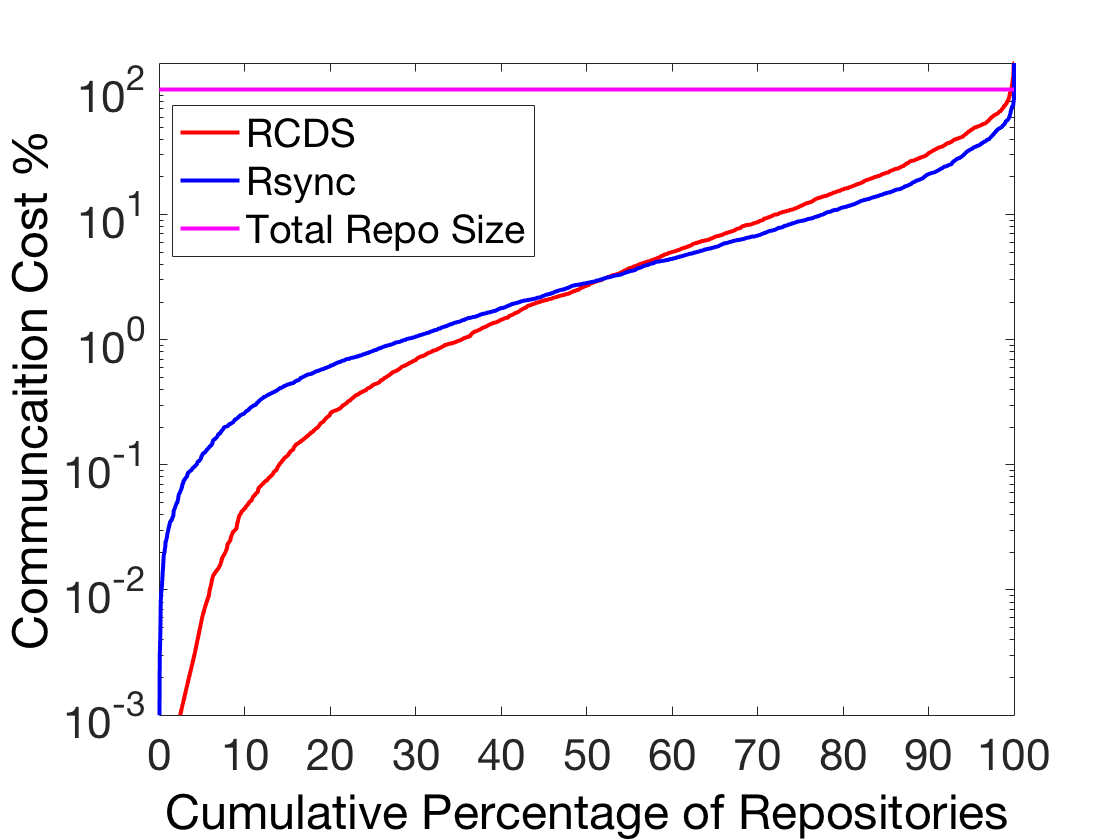}
        \caption{Synchronizing git repositories.}
        \label{fig:repo_perf}
    \end{subfigure}
    \begin{subfigure}{0.32\textwidth}
        \centering
        \includegraphics[width=\linewidth]{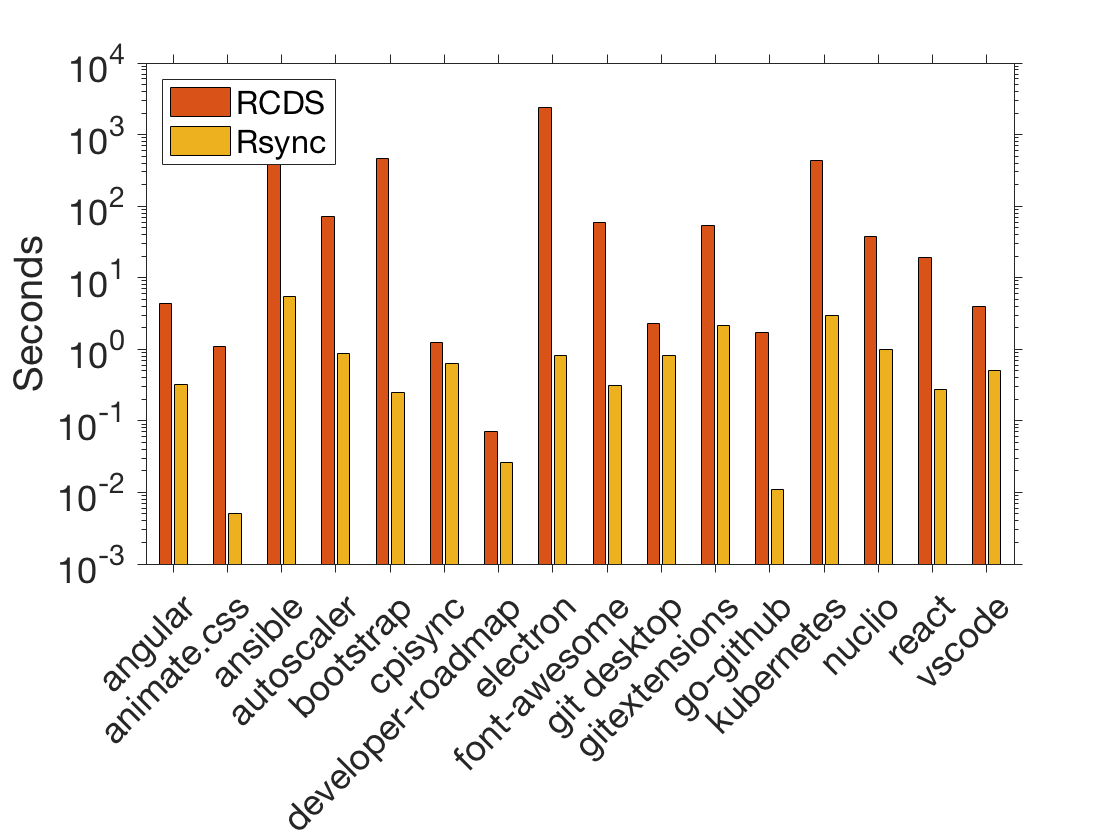}
        \caption{Time for synchronizing git repositories.}
        \label{fig:foldersync_time}
    \end{subfigure}
    \caption{Comparing performance of RCDS and \emph{rsync} synchronizing single files and git repositories.}
    \label{fig:all_repo_perf}
\end{figure*}

We extend to different burst edit distances and use a linear polynomial surface model to fit the planes of communication costs. The R-Square values are 0.9645 for RCDS and 0.9971 for \emph{rsync}. We calculate the equilibrium line for their total communication as
\begin{equation}
    y = 1.71x-17.17,
\end{equation} where $y$ is the edit burst distance and $x$ is the input string size. Given an input string size, RCDS performs better than \emph{rsync} if the edit burst distance is less than that of the equilibrium.

\subsection{GitHub Repositories}
Gitlab \cite{gitlab} is one of the widely used DevOps lifecycle tools that uses \emph{rsync} to manage repositories between different versions and states. We compare the performance characteristics of our protocol to that of the \emph{rsync} utility~\cite{rsyncWeb}. We utilize a basic heuristic to accommodate folder synchronization by checking file name and file size to determine if a file is different. Particularly, we hash each file name concatenating with its file size to a 64-bit hash value and use \emph{Interactive CPI} to reconcile the file name and size set. The newly created files are transferred in their entirety and files removed from the updated version are deleted. We also check the list of files that RCDS considers different with that of \emph{rsync} to ensure a fair comparison.


In \Cref{fig:repo_perf,fig:foldersync_time}, we present the communication and example time cost synchronizing the second latest to the latest release version of 5000 top starred public repositories, as of April 17th, 2019. In \Cref{fig:repo_perf}, the y-axis shows the percentage of the total repository size communicated for RCDS and \emph{rsync} to synchronize from the previous to the most recent version of repositories. The x-axis shows the percentile of repositories that require that communication cost or less. The graph shows that RCDS requires less communication than \emph{rsync} in 51\% of our tested repositories. For more than half of the repositories, RCDS performs significantly better than \emph{rsync} and for the rest \emph{rsync} only slightly outperforms RCDS.

\section{Conclusion}
\label{sec:conclusions}
We have presented a new string reconciliation protocol RCDS that exhibits both sub-linear communication complexity for many strings, and a computational complexity that is scalable to long strings.  In this way we have combined two features, which are not simultaneously available in existing solutions, toward a practical solution to this fundamental problem.  Our approach, which leverages context-based shingling and nearly optimal set reconciliation, improves upon the current standard - the heavily optimized \emph{rsync} utility - in a variety of scenarios.  Indeed, we have shown that reconciliation of different versions of popular git repositories can generally be completed with less communication using RCDS than with \emph{rsync}, and sometimes markedly so, although the converse is true for updates that involve large numbers of sparsely clustered edits.



We believe that our protocol mapply apply beyond incremental file updates to a wide variety of applications dependent upon presenting a consistent version of asynchronously edited ordered data, including distributed content distribution and metadata maintenance in distributed protocols.

\subsection{Future Work}
Our protocol utilizes many parameters, which can be tuned to specific applications.  These include the sizes and types of data, partitioning method, the underlying data structures, and set reconciliation protocol.  In addition, our partition tree currently fixes the maximum partition level, for ease of analysis, but one can similarly envision fixing the minimum length of a terminal string.  Furthermore, it would be interesting to extend our partition trees to handle incremental edits in the continuous file synchronization paradigm~\cite{wu2000continuous}. Finally, we envision RCDS being used in a hybrid scheme, which chooses which string reconciliation algorithm to run based on some initial, cheap assessment of the mutual properties of the strings involved.

\section*{Acknowledgments}
The authors would like to thank Sean Brandenburg for his assistance. This research is, in part, supported by the US National Science Foundation under Grants CCF-1563753.

\bibliographystyle{IEEEtran}
\bibliography{root}
%



\end{document}